\documentclass[11pt]{article}

\usepackage[margin=1in]{geometry}
\usepackage{amsmath,amssymb}

\usepackage{amsthm}
\usepackage{newtxtext,newtxmath}
\usepackage{graphicx}
\usepackage{booktabs}
\usepackage{placeins}
\usepackage{fontawesome5}
\usepackage{xcolor}
\definecolor{projectlink}{RGB}{28,76,120}
\usepackage[numbers,sort&compress]{natbib}
\usepackage{hyperref}
\usepackage{doi}
\hypersetup{
  colorlinks=true,
  linkcolor=black,
  citecolor=projectlink,
  urlcolor=projectlink,
  pdftitle={Sorting as Gradient Flow on the Permutohedron},
  pdfauthor={Jonathan Robert Landers},
  pdfkeywords={sorting algorithms, permutohedron, gradient flow, computational complexity}
}

\theoremstyle{plain}
\newtheorem{theorem}{Theorem}[section]
\newtheorem{proposition}[theorem]{Proposition}
\newtheorem{lemma}[theorem]{Lemma}
\newtheorem{corollary}[theorem]{Corollary}

\theoremstyle{definition}
\newtheorem{definition}[theorem]{Definition}

\newenvironment{acknowledgment}[1][Acknowledgment]{%
  \section*{#1}
  \addcontentsline{toc}{section}{#1}
  \normalsize
}{}

\begin{document}

\title{Sorting as Gradient~Flow on the Permutohedron}
% \author{Double Blind Submission}  % Replace with real names upon acceptance

\author{Jonathan Robert Landers\\
\texttt{jonathan.robert.landers@gmail.com}}
\date{}

\maketitle

\begin{abstract}
We investigate how sorting algorithms navigate the complexity of permutation space. Our main contribution is a continuous-time geometric model that casts sorting as directed motion on the permutohedron. A quadratic potential generates an ambient gradient flow that contracts toward the fixed sorted vertex $v_s=(1,2,\ldots,n)$. This dynamical picture is set against two discrete descriptions of the same problem. One follows adjacent-swap paths along the $1$-skeleton, while the other uses comparison half-spaces to refine the feasible order types. Together, they provide the combinatorial foils used to evaluate the continuous trajectory. Comparisons remove informational ambiguity, whereas the flow removes metric distance. The two mechanisms are complementary. To support this analysis, we present decision-tree arguments and local paths with $\Theta(n^2)$ behavior. We also show that the quadratic potential decreases strictly under every inversion-removing adjacent swap and formulate global half-space constraints on candidate rank maps. The maximal fixed-threshold relaxation time is determined exactly by the Euclidean diameter of the permutohedron, and the product of this relaxation time with the dimension is $\Theta(n\log n)$. Under a normalized comparison clock, this intrinsic geometric quantity has the same asymptotic scale as classical optimal comparison sorting.
\end{abstract}

\begin{center}
\href{https://github.com/jonland82/sorting-as-gradient-flow}{\faGithub\enspace GitHub repository}
\hspace{2.5em}
\href{https://jonland82.github.io/sorting-as-gradient-flow/}{\faGlobeAmericas\enspace Project website}
\end{center}

\section{Introduction}

Consider the fundamental problem of sorting a list
\begin{equation}
L = [\ell_1, \ell_2, \dots, \ell_n].
\end{equation}
The elements are taken to lie in a totally ordered set so that their relative order is well-defined. Rather than viewing sorting merely as a sequence of symbolic operations, it is useful to regard it as navigation through a large structured space. Each input determines one of $n!$ possible order types, and comparisons locate the correct one. Computational models then reveal distinct notions of motion, distance, and progress within that space.

\paragraph{Notation.}
The key issue is separating an item's original identity from its current position.
We therefore keep four related objects separate. The integer \(i\) is the original index of an item, and \(\ell_i\) is the value at that index. Since the entries are assumed pairwise distinct, define the original rank map
\begin{equation}
\rho(i)=\operatorname{rank}(\ell_i),
\end{equation}
where rank \(1\) is smallest and rank \(n\) is largest. At algorithmic time \(t\), let \(p_t(k)\) be the original index of the item currently occupying array position \(k\). The rank word plotted on the permutohedron is
\begin{equation}
x_t(k)=\rho(p_t(k)).
\end{equation}
Thus \(x_t\) records the ranks currently sitting in array positions, and initially
\begin{equation}
x_0=(\rho(1),\rho(2),\ldots,\rho(n)).
\end{equation}
The rank word is used here as an ex post representation of the input order type.
Sorting moves the rank word \(x_0\) toward \(v_s=(1,2,\ldots,n)\). The sorted index
order is a different object. Here \(\sigma(k)\) is the original index of the item
occupying sorted position \(k\), so \(\sigma=\rho^{-1}\), and
\begin{equation}
[\ell_{\sigma(1)},\ell_{\sigma(2)},\ldots,\ell_{\sigma(n)}]
\end{equation}
is sorted. For example, if \(L=[3,1,2]\), then \(\rho=(3,1,2)\) but
\(\sigma=(2,3,1)\). These differ in general. Table~\ref{tab:index_trace} later
traces all of the notation through the running example \(L=[5,4,2]\).

We write \(\tau\) for a permutation of the index set \(\{1,2,\dots,n\}\), acting on \(L\) by
\begin{equation}
\tau(L) = [\,\ell_{\tau(1)},\, \ell_{\tau(2)},\, \dots,\, \ell_{\tau(n)}\,].
\end{equation}

\paragraph{Setting.}
A brute-force baseline is to enumerate all \(n!\) permutations \(\tau\in S_n\) and select \(\sigma\), the permutation that produces the sorted list. This costs \(\Theta(n\cdot n!)\), since each candidate ordering requires \(O(n)\) time to check. We use standard asymptotic notation as in \cite{Knuth1976BigO,CLRS}. Bounds are worst-case unless stated otherwise, and randomized costs are denoted by \(\mathbb{E}[T(n)]\). Decision-tree logarithms are taken in base~2, while natural logarithms are used in exact continuous-time formulas. The choice of base does not affect asymptotic orders.

Within this framework, established sorting algorithms such as MergeSort, HeapSort, and QuickSort achieve a remarkable improvement, reducing the search from factorial to log-linear scaling. By contrast, local-swap methods such as BubbleSort use only adjacent comparisons and swaps, yielding quadratic $\Theta(n^2)$ behavior. MergeSort and HeapSort guarantee $O(n \log n)$ worst-case performance, while QuickSort attains the same bound in expectation. Each exploits regularities in the input, effectively pruning large regions of permutation space.

This striking efficiency has long inspired curiosity about the connection between organization and computational tractability. Foundational work by Knuth~\cite{Knuth}, Shannon~\cite{Shannon}, and others, along with recent combinatorial and geometric perspectives~\cite{blelloch2023, harris2024, Ziegler}, reflects ongoing efforts to understand how structure constrains and organizes computation.

\paragraph{Main contribution.}
We place three views of sorting in a common permutohedral state space. Adjacent swaps traverse edges, comparison half-spaces eliminate candidate rank maps, and a quadratic potential generates an exactly solvable ambient gradient flow toward the sorted vertex. The potential decreases strictly under every inversion-removing adjacent swap. For the continuous model, Theorem~\ref{thm:contraction} gives the straight-line contraction, Proposition~\ref{prop:maximal_relaxation} determines the exact maximal fixed-threshold relaxation time from the diameter of \(\mathcal{P}_n\), and Theorem~\ref{thm:dimension_relaxation} proves the intrinsic factorization \(\dim(\mathcal{P}_n)T_\varepsilon^{\max}(n)=\Theta(n\log n)\) for fixed \(\varepsilon>0\). Corollary~\ref{cor:normalized_clock} compares this independently proved geometric quantity with the normalized classical comparison scale. The decision-tree argument remains the rigorous computational source of the comparison lower bound; the flow supplies a geometric correspondence at the same asymptotic scale.

We begin by establishing the decision-tree complexity scale, then reinterpret comparisons as linear constraints acting on the permutohedron, and finally introduce a gradient-flow model that renders geometric contraction concrete.

\section{Information-Theoretic Decision-Tree Lower Bound}

The information-theoretic limit for comparison sorting follows by modeling the
process as a binary decision tree, whose leaves distinguish the possible
strict order types, equivalently the required output permutations.

\begin{lemma}[Decision-Tree Lower Bound]
  Let $T$ be a binary decision tree that correctly sorts $n$ elements using comparisons. Then the height $h$ of $T$ satisfies
  \begin{equation}
  h \ge \log_2(n!).
  \end{equation}
\end{lemma}

\begin{proof}
  A deterministic comparison sort must distinguish all $n!$ strict input orderings,
  and distinct order types cannot share a leaf, which prescribes a single output
  permutation. Thus a tree of height $h$ must satisfy
  \begin{equation}
  2^h \ge n!.
  \end{equation}
  Taking logarithms and using Stirling's approximation gives
  \begin{equation}
  h\ge\log_2(n!)=n\log_2 n-O(n)=\Omega(n\log n).
  \end{equation}
  Classical comparison sorts attain the matching upper bound. \qedhere
\end{proof}

This information-theoretic picture also has an enumerative counterpart. Harris,
Kretschmann, and Martínez Mori~\cite{harris2024} have shown
that the total number of comparisons made by QuickSort over all \(n!\) input
orderings coincides with the number of parking-preference lists admitting exactly
\(n-1\) lucky cars. This identity further underscores how comparison structure
governs the complexity of sorting.

Figure~\ref{fig:decision_tree} illustrates the decision tree for sorting three
elements, instantiated with the input
$L = (\ell_1 = 5, \ell_2 = 4, \ell_3 = 2)$. The leaves are labeled by candidate
sorting permutations of the original indices, with the reordered list shown for each
outcome. A highlighted path records the three comparisons performed on this
input and terminates at \(\sigma=(3,2,1)\),
which yields the sorted output $[2,4,5]$. Thus the example realizes
\(\lceil\log_2(3!)\rceil=3\), in agreement with the lower bound. Knuth's seminal
analysis established this information-theoretic limit for sorting algorithms
\cite{Knuth}.

\begin{figure}[!htb]
    \centering
    \includegraphics[width=0.9\linewidth,alt={Decision tree for sorting three elements, with a highlighted path leading to the sorted output.}]{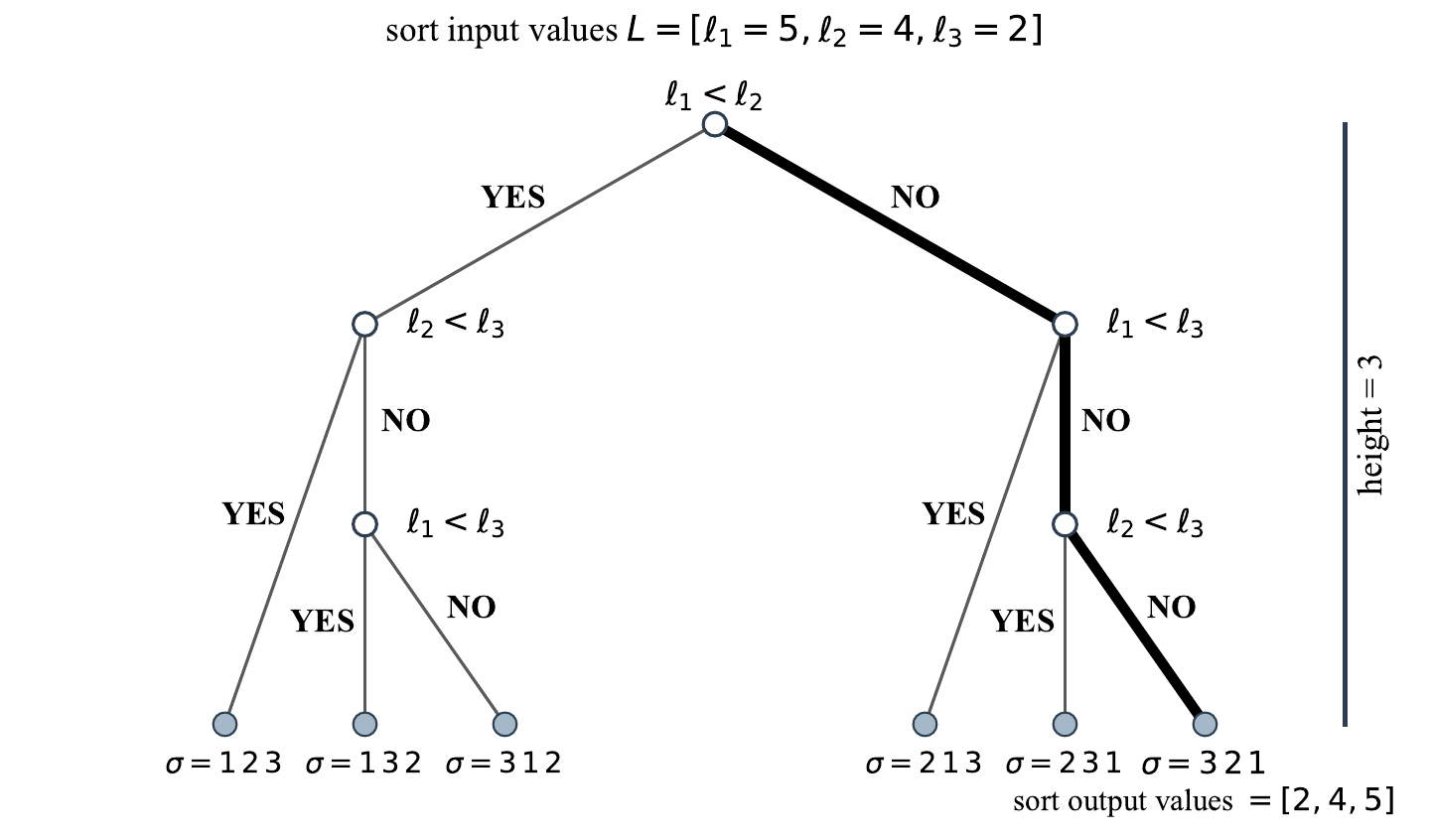}
    \caption{Decision tree for sorting
    \(L=(\ell_1=5,\ell_2=4,\ell_3=2)\). Internal nodes are comparisons and leaves are
    candidate sorting permutations. The highlighted path identifies
    \(\sigma=(3,2,1)\), yielding \([2,4,5]\) in three comparisons, matching
    \(\lceil\log_2(3!)\rceil=3\). Later permutohedron figures use rank-word coordinates
    rather than these index-order leaf labels.}
    \label{fig:decision_tree}
\end{figure}

The decision-tree model fixes the comparison complexity, but it does not explain the
geometry of the space being carved. This distinction matters. Algorithms restricted
to local adjacent comparisons, such as BubbleSort, resolve one inversion at a time
and have worst-case cost \(\Theta(n^2)\). Arbitrary comparisons can approach the
balanced \(\Theta(n\log n)\) bound. In the next section we express this contrast
on the permutohedron by interpreting comparison outcomes as linear inequalities.
Sorting then appears not only as logical branching but as structured
contraction of the candidate set until a single rank map remains.

\section{Geometric Perspective on the Permutohedron and Half-Space Constraints}

The same sorting process now admits two geometric readings on the permutohedron.
Adjacent swaps trace paths along its edges and produce the familiar $\Theta(n^2)$
behavior of local sorting algorithms. Comparisons act differently. As linear
inequalities, they remove incompatible order types from consideration. Thus one
geometry describes local motion, while the other captures information refinement
on the way to the optimal $\Theta(n\log n)$ comparison bound.

\begin{definition}[Permutohedron]
Let
\begin{equation}
v_s = (1, 2, \dots, n) \in \mathbb{R}^n.
\end{equation}
The permutohedron $\mathcal{P}_n$ is defined as
\begin{equation}
\mathcal{P}_n = \operatorname{conv}\{ \pi(v_s) \mid \pi \in S_n \}.
\end{equation}
Its vertex set is \(\mathcal{V}_n=\{ \pi(v_s) \mid \pi \in S_n\}\). In this paper these vertices are read as rank words. The coordinate in position \(k\) is the rank currently occupying that position.
\end{definition}

\begin{proposition}[Dimensionality and Affine Containment]
\label{prop:dimension}
The permutohedron $\mathcal{P}_n$ is an $(n-1)$-dimensional polytope contained in the affine hyperplane
\begin{equation}
H = \left\{ x \in \mathbb{R}^n \,\big|\, \sum_{i=1}^n x_i = \frac{n(n+1)}{2} \right\}.
\end{equation}
\end{proposition}

\begin{proof}
For any permutation $\pi \in S_n$,
\begin{equation}
\sum_{i=1}^n (\pi(v_s))_i = \sum_{i=1}^n i = \frac{n(n+1)}{2}.
\end{equation}
Thus, every vertex of $\mathcal{P}_n$ lies in the hyperplane $H$.

To see that these vertices span $H$, note that for each adjacent transposition $(i,i+1)$,
\begin{equation}
(i,i+1)(v_s) - v_s = \pm(e_{i+1}-e_i),
\end{equation}
with the sign depending only on the convention for writing the permutation action.
The vectors $e_{i+1} - e_i$ for $i = 1,\dots,n-1$ form a basis of the subspace
\begin{equation}
\left\{ x \in \mathbb{R}^n : \sum_{i=1}^n x_i = 0 \right\}.
\end{equation}
Since adjacent transpositions generate all permutations, it follows that the set
$\{\pi(v_s) : \pi \in S_n\}$ affinely spans the entire hyperplane $H$.

Because $H$ has dimension $n-1$, the permutohedron $\mathcal{P}_n$ is an $(n-1)$-dimensional polytope. \qedhere
\end{proof}

\subsection{Adjacent-Swap Walks and Their Geometric Interpretation}
Before introducing the linear-constraint formulation, it is helpful to visualize how
rank-word vertices relate to one another on the permutohedron $\mathcal{P}_n$. Because
adjacent transpositions generate $S_n$, they appear naturally as edges of the
$1$-skeleton (Cayley graph) of $\mathcal{P}_n$. Traversing these edges resolves one
inversion at a time. It gives a literal geometric realization of sorting as motion
on a polytope and mirrors the stepwise progression in the decision-tree example for
the input $L=(\ell_1,\ell_2,\ell_3)=(5,4,2)$.
Shortest paths from the identity to the reverse permutation in this graph are the
sorting networks studied by Angel et al.~\cite{angel2007}.
The walk diagrams use an inverse labeling of the standard permutohedron so that
exchanging neighboring array positions corresponds to traversing an edge. This
differs from direct coordinate labeling by permutation inversion.

Figure~\ref{fig:two_panel_permutohedron_walks} depicts this adjacent-swap path for
$n=3$ in two complementary forms. One is a flattened $2$-dimensional schematic and the other is a full
$3$-dimensional embedding of $\mathcal{P}_3$. In both views, the highlighted trajectory
connects the reverse rank word $(3,2,1)$ to the sorted rank word $(1,2,3)$ by
successive neighbor exchanges. These local moves provide a geometric model for the
$\Theta(n^{2})$ behavior of adjacent-swap sorting algorithms.
The low-dimensional example is used only because $\mathcal{P}_3$ can be drawn directly.
The definitions, constraints, and flow estimates below are stated for arbitrary $n$.

An edge path records how the current arrangement changes. A comparison history
instead determines which original rank maps remain possible. When its
outcomes are recorded as inequalities between candidate ranks of the original items,
they define half-spaces whose intersection with $\mathcal{V}_n$ narrows the search
space. Proposition~\ref{prop:geometric_constraints} connects this refinement to the
information-theoretic lower bound of $\Omega(n\log n)$ comparisons.

\begin{figure}[!htbp]
    \centering

    \begin{minipage}[t]{0.49\linewidth}
        \centering
        \includegraphics[width=\linewidth,alt={Two-dimensional schematic of an adjacent-swap path on the permutohedron.}]{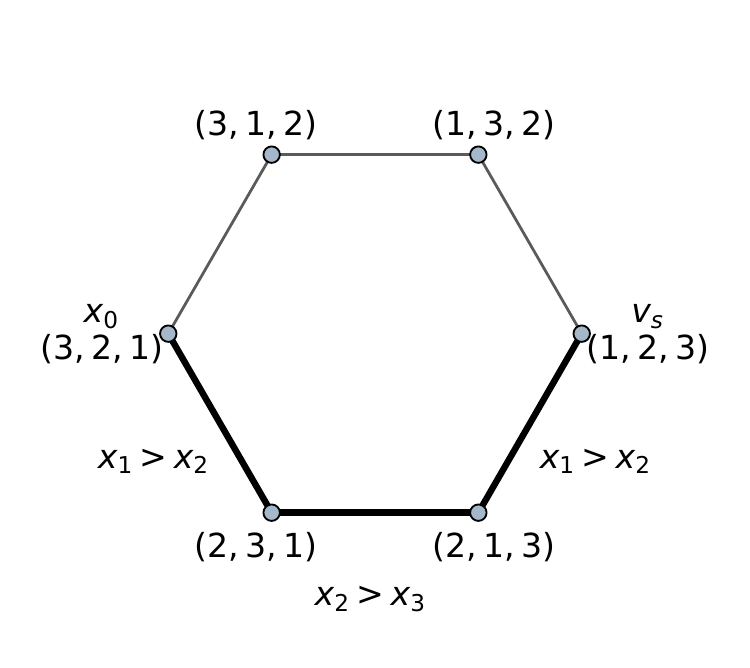}
    \end{minipage}
    \hfill
    \begin{minipage}[t]{0.49\linewidth}
        \centering
        \includegraphics[width=\linewidth,alt={Three-dimensional permutohedron with the adjacent-swap rank-word path from 321 to 123.}]{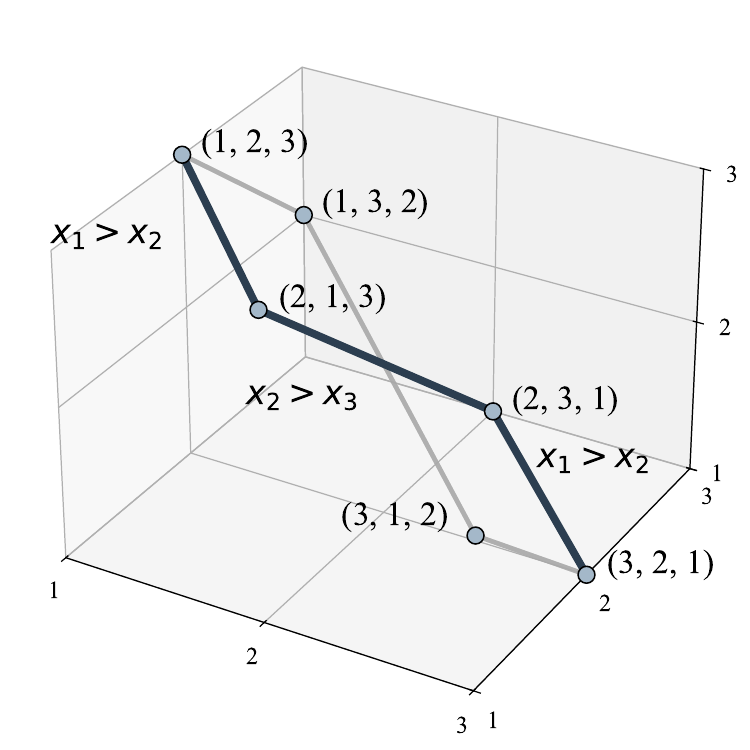}
    \end{minipage}

    \caption{Adjacent-position-swap motion on an inverse-labeled permutohedron \(\mathcal{P}_3\).
    \textbf{Left panel.} A 2D schematic of local traversal along the \(1\)-skeleton.
    \textbf{Right panel.} The same array-rank path
    \((3,2,1)\!\to\!(2,3,1)\!\to\!(2,1,3)\!\to\!(1,2,3)\)
    in a 3D embedding. Each step exchanges neighboring positions, illustrating the
    local \(\Theta(n^2)\) geometry of adjacent-swap sorting.}

    \label{fig:two_panel_permutohedron_walks}
\end{figure}

\subsection{Linear Constraints and Information Refinement}
Comparisons at different stages must be placed in a common coordinate frame. A
comparison is made between two current positions, but the information it reveals
concerns the two original items occupying those positions. To combine all comparison
outcomes into one fixed feasible set, the inequalities must therefore be recorded in
original-item coordinates.

Let
\begin{equation}
r=(r_1,\ldots,r_n)\in\mathcal{V}_n
\end{equation}
be a candidate original-index rank map, where \(r_i\) is the candidate rank of the
original item \(\ell_i\). If at time \(t\) a comparison between positions \(a\) and
\(b\) has the ``less than'' outcome, it records
\begin{equation}
r_{p_t(a)}<r_{p_t(b)}.
\end{equation}
In plain language, if position \(a\) currently contains item \(i\) and position \(b\)
contains item \(j\), the comparison records \(r_i<r_j\).

For a fixed comparison history \(\mathcal{C}\), let \(\mathrel{\triangleleft_j}\) be
\(<\) or \(>\) according to the outcome of comparison \(j\), made at time \(t_j\)
between positions \(a_j\) and \(b_j\). Its feasible candidate rank maps are
\begin{equation}
\mathcal{R}_n(\mathcal{C})
=
\left\{
r\in\mathcal{V}_n:
r_{p_{t_j}(a_j)}
\mathrel{\triangleleft_j}
r_{p_{t_j}(b_j)}
\text{ for every comparison }j
\right\}.
\end{equation}

\begin{table}[!htb]
  \centering
  \small
  \setlength{\tabcolsep}{4.5pt}
  \renewcommand{\arraystretch}{1.12}
  \begin{tabular}{@{}cccccc@{}}
    \toprule
    \(t\) & \(p_t\) & Current list & \(x_t=\rho\circ p_t\) & New constraint
      & \(\lvert\mathcal{R}_3(\mathcal{C}_t)\rvert\) \\
    \midrule
    \(0\) & \((1,2,3)\) & \([5,4,2]\) & \((3,2,1)\) & none & \(6\) \\
    \(1\) & \((2,1,3)\) & \([4,5,2]\) & \((2,3,1)\) & \(r_1>r_2\) & \(3\) \\
    \(2\) & \((2,3,1)\) & \([4,2,5]\) & \((2,1,3)\) & \(r_1>r_3\) & \(2\) \\
    \(3\) & \((3,2,1)\) & \([2,4,5]\) & \((1,2,3)\) & \(r_2>r_3\) & \(1\) \\
    \bottomrule
  \end{tabular}
  \caption{Index bookkeeping for adjacent-swap sorting of \(L=[5,4,2]\).
  Here \(\rho=\sigma=(3,2,1)\) only because the reverse permutation is
  self-inverse. In general \(\sigma=\rho^{-1}\) need not equal \(\rho\).}
  \label{tab:index_trace}
\end{table}
\FloatBarrier

Table~\ref{tab:index_trace} places the moving and fixed-coordinate descriptions
side by side. Row \(t\) records the state after the first \(t\) comparisons and
any resulting swaps, and \(\mathcal{C}_t\) denotes the outcomes accumulated by
that time. The left columns track motion. Here \(p_t\) records which original items
occupy the array positions, and \(x_t=\rho\circ p_t\) is the derived rank-word
coordinate. The right columns track information in fixed original-item
coordinates. Thus the rank word moves \(321\to231\to213\to123\) while the
feasible candidate set shrinks \(6\to3\to2\to1\). In the final row,
\(\mathcal{R}_3(\mathcal{C}_3)=\{\rho\}\) and \(x_3=v_s\). The general statement
is as follows.

\begin{proposition}[Comparison Constraints and Candidate Rank Maps]
\label{prop:geometric_constraints}
For an input with original rank map \(\rho\), a complete comparison history
\(\mathcal{C}\) at a decision-tree leaf identifies that map.
\begin{equation}
\mathcal{R}_n(\mathcal{C})=\{\rho\}.
\end{equation}
Consequently, a correct comparison decision tree has worst-case depth at least
\(\log_2(n!)=\Omega(n\log n)\), while optimal comparison sorts supply the matching
\(O(n\log n)\) upper bound. If \(\sigma=\rho^{-1}\) is the output permutation, then
applying it converts the identified rank map to the sorted rank word.
\begin{equation}
\bigl(\rho(\sigma(1)),\ldots,\rho(\sigma(n))\bigr)
=(1,\ldots,n)=v_s.
\end{equation}
\end{proposition}

\begin{proof}
Along a fixed decision-tree path, the positions compared and the original items
occupying them are determined by the preceding outcomes and operations. A candidate
rank map \(r\) follows this path exactly when it satisfies every recorded inequality,
so \(\mathcal{R}_n(\mathcal{C})\) is precisely the set of order types compatible with
the history.

At a leaf, the algorithm has fixed one output permutation. Correctness requires that
this permutation sort every candidate reaching that leaf, but an original-index rank
map \(r\) is sorted by the unique permutation \(r^{-1}\). Thus only one candidate rank
map can reach a correct leaf. Since the true map \(\rho\) follows the recorded path,
\(\mathcal{R}_n(\mathcal{C})=\{\rho\}\), and applying \(\sigma=\rho^{-1}\) gives the
displayed sorted rank word.

In the decision-tree model, each comparison has only two possible outcomes. Thus, at best,
it can distinguish between two classes of remaining candidate rank words. For the
purpose of proving a lower bound, we may therefore use the optimally balanced decision-tree
ideal. In this ideal, each comparison splits the remaining feasible candidates as evenly as possible.
This is the most optimistic binary refinement available to any comparison-based algorithm.
Such a cut need not be exactly balanced, and arbitrary comparison hyperplanes
may fail to halve either the vertices or the volume of the permutohedron.

In this ideal halving model, after
$k$ comparisons,
\begin{equation}
\frac{n!}{2^k} \le 1 \quad \Rightarrow \quad k \ge \log_2(n!) = \Omega(n\log n).
\end{equation}
Classical algorithms achieve $O(n\log n)$ comparisons. Hence their constraints cut
down the candidate original-index rank maps to exactly one feasible vertex, supplying
the matching upper bound.
\qedhere
\end{proof}

Table~\ref{tab:index_trace} records a literal execution. Figure~\ref{fig:permutohedron_sequence}
instead gives a canonical fixed-coordinate illustration of the decision-tree halving
heuristic. The inequalities \(x_1<x_2\)
and \(x_2<x_3\) isolate the increasing chamber and hence the vertex \((1,2,3)\).
For a concrete execution, comparison outcomes are first
recorded in original-item coordinates and are then converted to sorted-output coordinates.
Under this approximation, the panels represent exponential
reduction in the number of feasible permutations. An actual comparison may not split
the remaining vertices or polytope volume evenly.

\begin{figure}[!htbp]
    \centering
    \includegraphics[width=\linewidth,alt={Canonical permutohedron panels in which the fixed-coordinate cuts x1 less than x2 and x2 less than x3 isolate the increasing rank word.}]{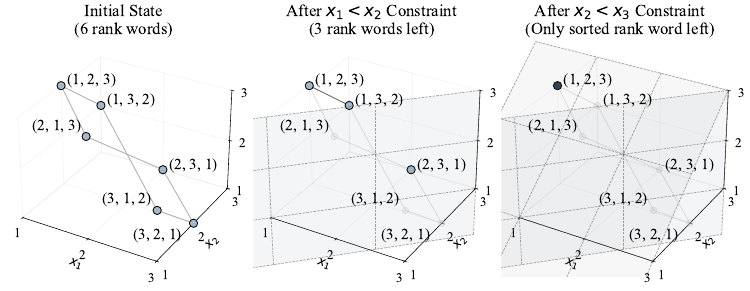}
    \caption{Canonical comparison half-spaces on \(\mathcal{P}_3\). The cuts
    \(x_1<x_2\) and \(x_2<x_3\) isolate the increasing rank word. The picture
    illustrates chamber refinement rather than the literal adaptive comparison path
    of a particular input.}
    \label{fig:permutohedron_sequence}
\end{figure}

Several nearby approaches also connect sorting with geometry. Blelloch and Dobson
relate comparison sorting to offline binary search trees through planar permutation
representations and the log-interleave bound~\cite{blelloch2023}. Lee et al.\ study
stack-sorting simplices, subpolytopes encoding the geometry of a constrained sorting
process~\cite{lee2025}. Continuous relaxations include sorting-network formulations
of permutation problems~\cite{lim2014}, entropy-regularized optimal-transport
sorting~\cite{cuturi2019}, and differentiable sorting by projection onto the
permutahedron~\cite{blondel2020}. These approaches provide geometric models of
particular sorting structures and optimization schemes.

The three viewpoints now on the table use distinct states and measures of
progress. With adjacent exchanges, this is graph distance on the \(1\)-skeleton of \(\mathcal{P}_n\),
equivalently the number of inversions to be removed. With arbitrary comparisons, the
state is the set of feasible original-index rank maps, and disorder is the information
still needed to isolate one of them. Euclidean distance in the standard embedding is a
third object.

The central comparison therefore separates two geometrically different processes. First,
comparisons eliminate incompatible order types until a single permutation remains.
Second, once the target permutation is represented as a vertex, Euclidean gradient
flow gives a continuous model of contraction toward it. These mechanisms are not
identical. One is informational and the other metric. After comparisons are normalized
by \(n\), both exhibit logarithmic behavior. We now turn to the metric side of this
picture.

\section[\texorpdfstring{Continuous Gradient Flow and the Classical $n\log n$ Scale}
         {Continuous Gradient Flow and the Classical n log n Scale}]
        {Continuous Gradient Flow and the Classical $n\log n$ Scale}

Starting from a fixed target, Theorem~\ref{thm:contraction} gives the ambient gradient
flow in closed form and quantifies its contraction time.

The closest dynamical antecedent is Brockett's double-bracket flow, which sorts a
list through gradient-driven matrix diagonalization on an isospectral orbit whose
moment image is the permutohedron~\cite{brockett1991,helmke1994}. By contrast, the
model here works directly
in rank-word coordinates on the permutohedron and places fixed-target contraction
alongside adjacent-swap motion and comparison half-spaces.

\subsection{Setup and Definitions}

Let
\begin{equation}
v_s = (1,2,\dots,n) \in \mathbb{R}^n
\end{equation}
denote the vertex of the permutohedron corresponding to the sorted rank word.
For the continuous relaxation, work in the standard Euclidean embedding of the
permutohedron and define the potential
\begin{equation}
V(x) = \frac{1}{2}\|x-v_s\|_2^2 = \frac{1}{2}\sum_{i=1}^n (x_i-i)^2.
\end{equation}
This is a Spearman-type squared rank displacement. It is not the only sorting metric.
Adjacent swaps use inversion distance, whereas comparisons use the information
remaining in the set of candidate rank maps. The role of \(V\) is to
provide a smooth benchmark on \(\mathcal{P}_n\). Note that \(V(v_s)=0\) and
\(V(x) > 0\) for \(x \neq v_s\). Given an input list \(L\)
with original rank map \(\rho\), we take as initial state its rank word
\begin{equation}
x(0)=x_0=(\rho(1),\rho(2),\ldots,\rho(n)) \in \mathcal{P}_n,
\end{equation}
so that the relaxation carries this fixed geometric representative toward the
sorted vertex \(v_s\).

\begin{proposition}[Adjacent-swap descent]
\label{prop:adjacent_swap_descent}
Let \(x\in\mathcal{V}_n\), and suppose \(x_i=a>b=x_{i+1}\). If \(x'\) is obtained
by swapping these adjacent inverted ranks, then
\begin{equation}
V(x)-V(x')=a-b>0.
\end{equation}
Consequently, \(V\) is a strict Lyapunov function for adjacent
inversion-removing sorting. If \(\operatorname{Inv}(x)\) denotes the inversion
number of \(x\), then
\begin{equation}
\operatorname{Inv}(x)\le V(x)\le (n-1)\operatorname{Inv}(x).
\end{equation}
\end{proposition}

\begin{proof}
Only coordinates \(i\) and \(i+1\) change, so
\begin{align}
V(x)-V(x')
&=\frac12\left[(a-i)^2+(b-i-1)^2
 -(b-i)^2-(a-i-1)^2\right] \\
&=a-b.
\end{align}
Every adjacent inversion-removing swap decreases the inversion number by exactly
one. Sorting \(x\) therefore takes \(\operatorname{Inv}(x)\) such swaps, and each
decrease in \(V\) lies between \(1\) and \(n-1\). Summing the decreases along the
path to \(v_s\), where \(V(v_s)=0\), gives the stated bounds.
\end{proof}

Thus the same potential measures descent along the discrete adjacent-swap walk and
the continuous flow introduced next. The two-sided bound parallels the
Diaconis--Graham inequality between the Spearman footrule and the inversion
count~\cite{diaconis1977}.

\subsection{Exact Gradient-Flow Contraction}

\begin{theorem}[Ambient Euclidean Contraction]
\label{thm:contraction}
Let \(x(0)\in \mathbb{R}^n\) be the initial rank word corresponding to an unsorted
input list and let \(v_s=(1,2,\dots,n)\) be the sorted vertex of the permutohedron \(\mathcal{P}_n\).
Define
\begin{equation}
D_0=\|x(0)-v_s\|_2^2.
\end{equation}
Under the gradient flow
\begin{equation}
\dot{x}(t)=-\nabla V(x(t))=v_s-x(t),
\end{equation}
the following statements hold.
\begin{enumerate}
    \item The solution satisfies
    \begin{equation}
    x(t)=v_s+\bigl(x(0)-v_s\bigr)e^{-t},
    \end{equation}
    \begin{equation}
    \|x(t)-v_s\|_2=\|x(0)-v_s\|_2e^{-t},
    \end{equation}
    and hence
    \begin{equation}
    \|x(t)-v_s\|_2^2=\|x(0)-v_s\|_2^2e^{-2t}.
    \end{equation}
    \item For a threshold \(0<\varepsilon<\sqrt{D_0}\), the condition
    \(\|x(t)-v_s\|_2\le \varepsilon\) holds once
    \begin{equation}
    t\ge \frac{1}{2}\ln\Bigl(\frac{D_0}{\varepsilon^2}\Bigr).
    \end{equation}
    \item For the reverse rank word \(x(0)=(n,n-1,\dots,1)\),
    \begin{equation}
    D_0=\sum_{i=1}^{n}(n+1-2i)^2=\frac{n(n^2-1)}{3}=\Theta(n^3),
    \end{equation}
    and therefore
    \begin{equation}
    V(x(0))=\frac{D_0}{2}=\frac{n(n^2-1)}{6}.
    \end{equation}
    For fixed \(\varepsilon>0\), the threshold time is \(t=\Theta(\ln n)\).
\end{enumerate}
\end{theorem}

\begin{proof}
We begin by deriving the gradient flow dynamics directly from the potential \(V\).
A straightforward computation shows that
\begin{equation}
\nabla V(x) = \bigl(x_1-1,\; x_2-2,\; \dots,\; x_n-n\bigr).
\end{equation}
Hence, the gradient flow is given by
\begin{equation}
\dot{x}(t) = -\nabla V(x(t)) =
\bigl(1-x_1(t),\; 2-x_2(t),\; \dots,\; n-x_n(t)\bigr).
\end{equation}
Because the system decouples, each coordinate \(x_i\) satisfies the first-order
linear differential equation
\begin{equation}
\dot{x}_i(t) = i-x_i(t),
\end{equation}
with initial condition \(x_i(0)\). Its solution is
\begin{equation}
x_i(t) = i+\bigl(x_i(0)-i\bigr)e^{-t}.
\end{equation}
Collecting these coordinate solutions, we obtain
\begin{equation}
x(t)=v_s+\bigl(x(0)-v_s\bigr)e^{-t},
\end{equation}
so that
\begin{equation}
\|x(t)-v_s\|_2^2=\|x(0)-v_s\|_2^2e^{-2t}.
\end{equation}
This establishes statement~1.

For statement~2, define the \emph{initial disorder} \(D_0=\|x(0)-v_s\|_2^2\). The
system is effectively sorted when
\begin{equation}
\|x(t)-v_s\|_2^2 = D_0e^{-2t}\le \varepsilon^2.
\end{equation}
Taking logarithms yields
\begin{equation}
t\ge \frac{1}{2}\ln\Bigl(\frac{D_0}{\varepsilon^2}\Bigr).
\end{equation}
\textbf{Reverse-permutation radius.} If \(x(0)=(n,n-1,\dots,1)\) then
\(x_i(0)-i=n+1-2i\). The standard centered-square identity gives
\begin{equation}
D_0=\sum_{i=1}^{n}\bigl(n+1-2i\bigr)^2
     = \frac{n(n^2-1)}{3}
     = \Theta(n^3).
\end{equation}
Consequently the initial potential at the reverse rank word is
\begin{equation}
V(x(0))=\frac12D_0=\frac{n(n^2-1)}{6}.
\end{equation}
Hence \(t=\Theta(\ln n)\).
\end{proof}

\subsection{Maximal Relaxation Time}

For \(x_0\in\mathcal{P}_n\), define the time to a Euclidean threshold
\begin{equation}
T_\varepsilon(x_0)=
\inf\bigl\{t\geq0:\|x(t)-v_s\|_2\leq\varepsilon\bigr\},
\end{equation}
and define the worst-case relaxation time by
\begin{equation}
T_\varepsilon^{\max}(n)=\sup_{x_0\in\mathcal{P}_n}T_\varepsilon(x_0).
\end{equation}
Theorem~\ref{thm:contraction} gives, including the case in which the initial point
is already within the threshold,
\begin{equation}
T_\varepsilon(x_0)=
\begin{cases}
0, & \|x_0-v_s\|_2\leq\varepsilon,\\[0.25ex]
\displaystyle\ln\dfrac{\|x_0-v_s\|_2}{\varepsilon},
& \|x_0-v_s\|_2>\varepsilon.
\end{cases}
\end{equation}

\begin{proposition}[Exact Maximal Relaxation Time]
\label{prop:maximal_relaxation}
Let \(0<\varepsilon<\operatorname{diam}(\mathcal{P}_n)\). Under the flow of
Theorem~\ref{thm:contraction},
\begin{equation}
T_\varepsilon^{\max}(n)
=\ln\frac{\operatorname{diam}(\mathcal{P}_n)}{\varepsilon}
=\frac12\ln\left(\frac{n(n^2-1)}{3\varepsilon^2}\right).
\end{equation}
In particular, for every fixed \(\varepsilon>0\), as \(n\to\infty\),
\begin{equation}
T_\varepsilon^{\max}(n)=\frac32\ln n+O(1).
\end{equation}
\end{proposition}

\begin{proof}
The unsquared-distance identity in Theorem~\ref{thm:contraction} gives
\begin{equation}
\|x(t)-v_s\|_2=\|x_0-v_s\|_2e^{-t}.
\end{equation}
Thus, whenever \(\|x_0-v_s\|_2>\varepsilon\),
\begin{equation}
T_\varepsilon(x_0)=\ln\frac{\|x_0-v_s\|_2}{\varepsilon}.
\end{equation}
It remains to maximize the initial distance. The function
\(x\mapsto\|x-v_s\|_2^2\) is convex, so its maximum over the compact polytope
\(\mathcal{P}_n\) is attained at a vertex. Indeed, writing a point as a convex
combination of vertices bounds its value by the largest value at those vertices.
For a vertex \(\pi(v_s)\),
\begin{equation}
\|\pi(v_s)-v_s\|_2^2
=2\sum_{i=1}^n i^2-2\sum_{i=1}^n i\,\pi(i).
\end{equation}
By the rearrangement inequality, the second sum is minimized by
\(\pi(i)=n+1-i\). Hence the reverse vertex is farthest from \(v_s\), at squared
distance
\begin{equation}
\sum_{i=1}^n(n+1-2i)^2=\frac{n(n^2-1)}{3}.
\end{equation}
The Euclidean diameter of a polytope is attained by two vertices. Coordinate
permutations act transitively and isometrically on the vertices of
\(\mathcal{P}_n\), so any such pair may be carried to \(v_s\) and another vertex.
Consequently,
\begin{equation}
\operatorname{diam}(\mathcal{P}_n)
=\left\|v_s-(n,n-1,\dots,1)\right\|_2
=\sqrt{\frac{n(n^2-1)}{3}}.
\end{equation}
Substitution proves the exact formula, and its fixed-\(\varepsilon\) expansion gives
the final assertion.
\end{proof}

\subsection{Geometric Factorization of the Classical Scale}

\begin{theorem}[Dimension--Relaxation Factorization]
\label{thm:dimension_relaxation}
For every fixed threshold \(\varepsilon>0\) and all sufficiently large \(n\),
\begin{equation}
\dim(\mathcal{P}_n)T_\varepsilon^{\max}(n)
=\frac{n-1}{2}\ln\left(\frac{n(n^2-1)}{3\varepsilon^2}\right).
\end{equation}
Consequently,
\begin{equation}
\dim(\mathcal{P}_n)T_\varepsilon^{\max}(n)
=\frac32n\ln n+O(n)=\Theta(n\log n).
\end{equation}
Equivalently,
\begin{equation}
\dim(\mathcal{P}_n)\ln\frac{\operatorname{diam}(\mathcal{P}_n)}{\varepsilon}
=\Theta(n\log n).
\end{equation}
\end{theorem}

\begin{proof}
Proposition~\ref{prop:dimension} gives \(\dim(\mathcal{P}_n)=n-1\), and
Proposition~\ref{prop:maximal_relaxation} gives the exact relaxation time. Their
product is the displayed identity. For fixed \(\varepsilon>0\),
\begin{equation}
\ln\left(\frac{n(n^2-1)}{3\varepsilon^2}\right)
=\ln\left(\frac{n^3}{3\varepsilon^2}
\left(1-\frac1{n^2}\right)\right)
=3\ln n+O(1).
\end{equation}
Multiplying by \((n-1)/2\) yields \(\frac32n\ln n+O(n)\).
\end{proof}

This theorem is an intrinsic geometric statement about the standard permutohedron
and its quadratic relaxation, proved entirely from its dimension, Euclidean diameter,
and exact contraction law. Corollary~\ref{cor:normalized_clock} next places this
geometric quantity alongside normalized optimal comparison cost.

Let \(k(n)\) denote the optimal worst-case number of comparisons required to sort
\(n\) distinct elements, and define the normalized comparison clock
\begin{equation}
s(n)=\frac{k(n)}{n}.
\end{equation}

\begin{corollary}[Normalized Comparison Clock]
\label{cor:normalized_clock}
The classical comparison bounds imply
\begin{equation}
s(n)=\Theta(\log n).
\end{equation}
Hence, for every fixed \(\varepsilon>0\),
\begin{equation}
s(n)=\Theta\bigl(T_\varepsilon^{\max}(n)\bigr).
\end{equation}
\end{corollary}

\begin{proof}
The classical lower and upper bounds give \(k(n)=\Theta(n\log n)\). Dividing by
\(n\) gives \(s(n)=\Theta(\log n)\), while
Proposition~\ref{prop:maximal_relaxation} gives
\(T_\varepsilon^{\max}(n)=\Theta(\log n)\) for fixed \(\varepsilon>0\).
\end{proof}

This comparison concerns asymptotic scales. It does not assert that one unit of flow
time is literally equivalent to one sweep of comparisons, and it does not replace
the decision-tree proof of the comparison lower bound.

\subsection{Constraint Geometry and the Ambient Flow}

The exact flow in Theorem~\ref{thm:contraction} is an ambient Euclidean benchmark.
It is already compatible with the convex geometry of the permutohedron.
\(\mathcal{P}_n\) lies in the affine hyperplane
\begin{equation}
\sum_{i=1}^{n}x_i=\frac{n(n+1)}{2},
\end{equation}
and both \(x(0)\) and \(v_s\) lie in \(\mathcal{P}_n\). Since \(\mathcal{P}_n\) is
convex, the straight segment
\begin{equation}
x(t)=v_s+\bigl(x(0)-v_s\bigr)e^{-t}
\end{equation}
remains in \(\mathcal{P}_n\) for all \(t\ge0\).

It is important not to identify comparison hyperplanes with the boundary facets of
\(\mathcal{P}_n\). The hyperplanes
\begin{equation}
x_i=x_j
\end{equation}
are braid-arrangement walls in the ambient affine space. They slice
\(\mathcal{P}_n\). In general, they are not facets of the permutohedron. The
facets of the standard permutohedron are instead described by subset-sum inequalities~\cite{Ziegler}
of the form
\begin{equation}
\sum_{i\in S}x_i \ge \frac{|S|(|S|+1)}{2},
\qquad
\varnothing\ne S\subsetneq \{1,\dots,n\},
\end{equation}
together with the fixed total-sum equality. The picture thus makes a subtlety
visible. Lying inside \(\mathcal{P}_n\) does not by itself force any projection of
the flow.

This distinction now becomes concrete. Comparison-feasible sets describe information,
whereas constraints on a metric trajectory would restrict motion.

\begin{proposition}[Projection is inactive when the target is feasible]
\label{prop:inactive_projection}
Let \(K\subseteq\mathcal{P}_n\) be a closed convex set containing \(v_s\). For every
\(x\in K\),
\begin{equation}
v_s-x\in T_K(x).
\end{equation}
Here \(T_K(x)\) is the tangent cone of \(K\) at \(x\). Therefore its Euclidean
projection satisfies
\begin{equation}
\Pi_{T_K(x)}(v_s-x)=v_s-x,
\end{equation}
and the projected flow agrees with the ambient flow.
\begin{equation}
\dot{x}=v_s-x.
\end{equation}
\end{proposition}

\begin{proof}
Since \(x,v_s\in K\) and \(K\) is convex,
\begin{equation}
x+\lambda(v_s-x)\in K
\qquad (0\leq\lambda\leq1).
\end{equation}
Thus \(v_s-x\) is a feasible tangent direction at \(x\), so projecting it onto
the tangent cone leaves it unchanged~\cite{rockafellar1970}.
\end{proof}

If a convex feasible region contains \(v_s\), projection is therefore inactive.
If instead the region excludes \(v_s\), constrained descent approaches the unique
point in the region nearest to \(v_s\), not \(v_s\) itself. The two objects thus play
distinct and complementary roles. Comparison half-spaces describe information
reduction, and the Euclidean flow describes metric contraction. Comparisons remove
informational ambiguity. Gradient flow removes metric distance.

\subsection{Remarks}

Theorem~\ref{thm:contraction} gives exact exponential decay of Euclidean rank
displacement under \(V(x)=\tfrac12\|x-v_s\|_2^2\). Proposition~\ref{prop:maximal_relaxation}
and Theorem~\ref{thm:dimension_relaxation} extract its exact worst-case time and
dimension--relaxation factorization, while Corollary~\ref{cor:normalized_clock}
compares that geometric scale with normalized optimal comparison cost. More broadly,
a descent formulation can clarify a process by
identifying an energy that decreases along its path. In this respect the
viewpoint echoes the insight of Jordan,
Kinderlehrer, and Otto, who reinterpreted the Fokker--Planck equation as a
gradient flow minimizing a free energy functional under the Wasserstein metric
\cite{jko1998}.

Figure~\ref{fig:gradient_flow_permutohedron} brings the discrete and continuous
descriptions together. Vertices represent rank words, and the gradient flow traces
a smooth curve inside the convex permutohedron toward the sorted vertex. The left
panel contrasts this motion with the adjacent--swap path. Local algorithms advance
by short moves along edges, while the continuous relaxation cuts straight across
the polytope. This path is a genuine cross-section, not a graph-geodesic
on the \(1\)-skeleton.

In scalar form, the right panel presents the dynamics. The distance to the
sorted state contracts exponentially as \(r(t)=\|x(t)-v_s\|_{2}\) decreases.
Viewed in \((r,t,V)\)-space, the trajectory resembles motion through a curved
landscape. The system rolls downhill under the potential, like a particle
descending a gravitational well toward equilibrium. In both panels the example
is \(L=[5,4,2]\), used earlier. The rank word begins at
\((3,2,1)\) and flows toward \((1,2,3)\), corresponding to the sorted output
\([2,4,5]\). The initial radius is
\(r_0=\|(3,2,1)-(1,2,3)\|_2=\sqrt{8}\), and the dissipation at \(t=0\) is
\(dV/dt=-\|\nabla V(x(0))\|_2^2=-8\).

Together, the panels keep the two roles visible. The left contrasts discrete and
continuous paths, while the right isolates exponential metric convergence. For an
initial radius above the threshold, the displayed time is
\(T_\varepsilon(x_0)=\ln(\|x_0-v_s\|_2/\varepsilon)\). Its logarithmic order agrees
with normalized optimal comparison cost.

\begin{figure}[!htbp]
    \centering

    \begin{minipage}[t]{0.49\linewidth}
        \centering
        \includegraphics[width=\linewidth,alt={Straight-line ambient gradient-flow path from the reverse rank word to the sorted rank word.}]{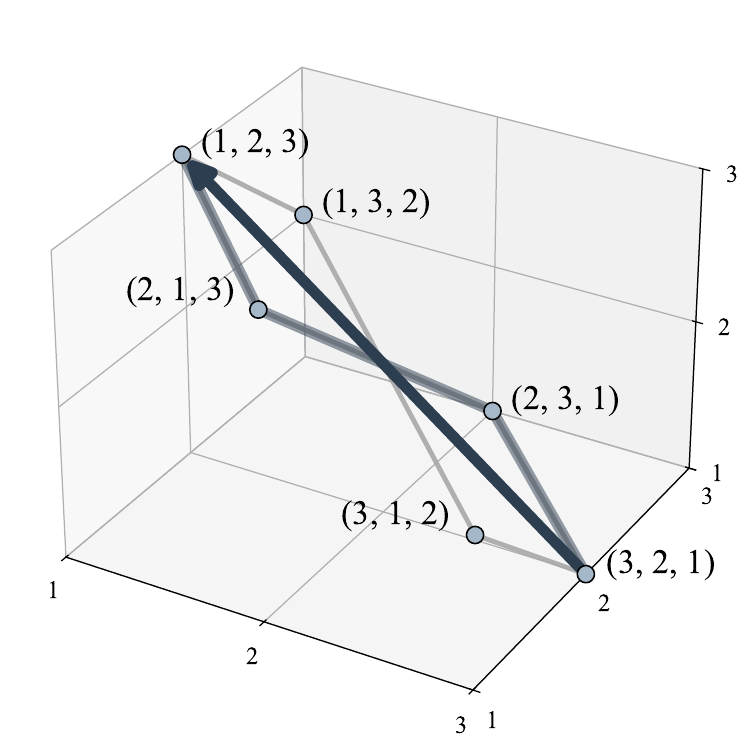}
    \end{minipage}
    \hfill
    \begin{minipage}[t]{0.49\linewidth}
        \centering
        \includegraphics[width=\linewidth,alt={Plot of exponential contraction of distance to the sorted vertex over time.}]{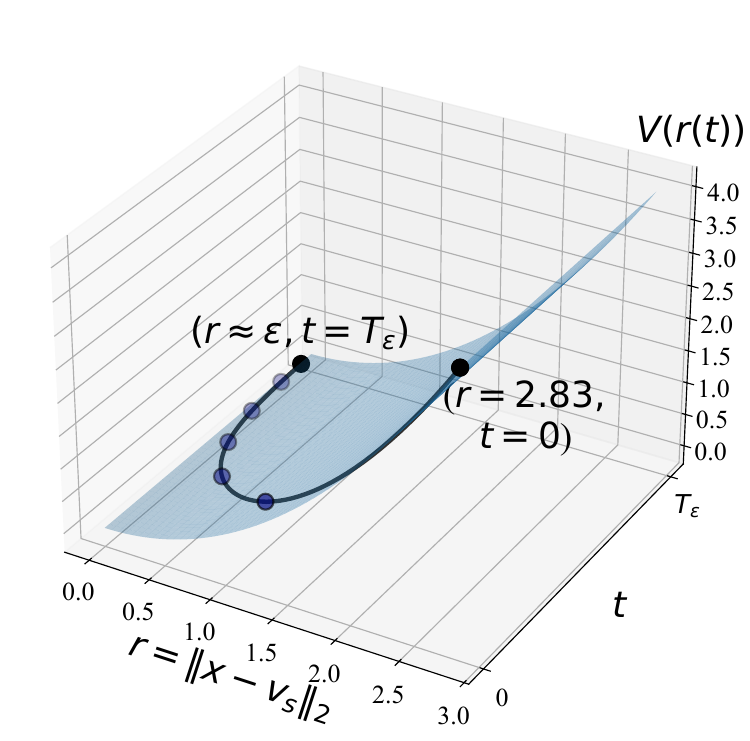}
    \end{minipage}
    \caption{Gradient-flow relaxation under
    \(V(x)=\tfrac12\|x-v_s\|_2^2\). \textbf{Left panel.} Straight-line descent on
    \(\mathcal{P}_3\) from \((3,2,1)\) to \((1,2,3)\), contrasted with the adjacent-swap
    path. \textbf{Right panel.} Exponential decay of
    \(r(t)=\|x(t)-v_s\|_2\) from \(r_0=\sqrt8\) to a threshold \(\varepsilon\).
    This Euclidean benchmark represents continuous metric relaxation, whose
    fixed-threshold time has the same asymptotic scale as normalized optimal comparison
    cost.}

    \label{fig:gradient_flow_permutohedron}
\end{figure}

\section{Conclusion}

Efficient algorithms overcome the factorial permutation space by exploiting structure.
From this perspective, the gain over local sorting appears as a contrast between local
motion and information. Adjacent-swap algorithms traverse the permutohedron
along its edges, resolving one inversion at a time and producing the classical
\(\Theta(n^2)\) behavior. Arbitrary comparisons act differently. They impose
global constraints on candidate rank maps, eliminating incompatible order types
and reaching the \(\Theta(n\log n)\) bound. Comparison information
therefore shrinks the feasible order-type set as a whole. This theme is consistent
with structured decomposition methods in combinatorial algorithms \cite{HopcroftTarjan}.
It also aligns with recent enumerative work of Harris, Kretschmann, and Mart{\'\i}nez Mori, where
QuickSort comparisons are counted by parking preference lists with \(n-1\)
``lucky cars'' \cite{harris2024}.

The continuous model adds the dynamical view. An ambient gradient flow on the
permutohedron represents sorting as a trajectory through space and clarifies the three
mechanisms at play. Adjacent swaps follow the edges of the \(1\)-skeleton, arbitrary
comparisons act through information constraints, and the Euclidean flow supplies a
smooth rank-displacement benchmark inside the same polytope. It reduces squared
distance to the sorted order according to
\begin{equation}
\|x(t)-v_s\|_2^2=\|x(0)-v_s\|_2^2e^{-2t}.
\end{equation}
Whenever \(0<\varepsilon<\operatorname{diam}(\mathcal{P}_n)\), the maximal
fixed-threshold relaxation time satisfies
\begin{equation}
T_\varepsilon^{\max}(n)
=\ln\frac{\operatorname{diam}(\mathcal{P}_n)}{\varepsilon},
\end{equation}
and therefore, for fixed \(\varepsilon>0\),
\begin{equation}
\dim(\mathcal{P}_n)T_\varepsilon^{\max}(n)=\Theta(n\log n).
\end{equation}
This is an intrinsic geometric factorization of the classical scale. The decision-tree
argument remains the computational proof of the comparison lower bound, while division
of optimal comparison cost by \(n\) gives the same logarithmic order as the relaxation.

Shannon's information-theoretic analysis made clear that search is governed by
available information \cite{Shannon}. Complexity theory develops this principle
across many computational settings \cite{arora2009}. Mann's perspective on search
emphasizes the way constraints organize rather than merely restrict a space
\cite{mann2017}. The present model gives this idea a concrete geometric form for
sorting. Factorial search collapses into log-linear time when comparisons supply
enough information to organize the candidate rank maps. In that sense, local traversal
and global structure meet inside the permutohedron, together with its comparison
half-spaces and continuous flow toward order. Informational
refinement and metric contraction thus emerge as distinct but complementary geometric
views of sorting. The permutohedron links them through the common asymptotic scale
encoded by its dimension and Euclidean diameter.

\begin{acknowledgment}
The author used software tools, including AI assistance, for figure preparation, algebraic checks, and editing. The author is responsible for all mathematical content and conclusions.
\end{acknowledgment}

% \begin{flushleft}
% \footnotesize
% \textbf{Disclosure of interest.} The author(s) report no conflicts of interest.\\[0.5em]
% \textbf{Funding.} No funding was received for this research.
% \end{flushleft}

\begin{flushleft}
\footnotesize
\textbf{Disclosure of interest.} The author reports no conflict of interest.\\[0.5em]
\textbf{Funding.} No funding was received for this research.
\end{flushleft}

\bibliographystyle{unsrtnat}
\bibliography{references}
\end{document}